\documentclass[aps,pra,twocolumn,showpacs,superscriptaddress]{revtex4-1}
\usepackage{amssymb,amsmath,amsfonts,bm,amsbsy,latexsym,amsthm,dsfont,array,layout,mathrsfs,mathdots,ucs}
\usepackage[normalem]{ulem}
\usepackage[usenames,dvipsnames]{color}
 
\newcommand{\ket}[1]{\left|#1\right\rangle}
\newcommand{\bra}[1]{\left\langle #1\right|}
\newcommand{\bracket}[2]{\left\langle #1|#2\right\rangle}
\newcommand\defn[1]{\textsl{#1}}

\newcommand\ketbra[2]{|#1\rangle\langle#2|}
\newcommand\cM{{\mathcal M}}
\newcommand\cN{{\mathcal N}}

\newcommand\cE{{\mathcal E}}

\newcommand{\proj}[1]{\ket{#1}\bra{#1}}

\newtheorem{lemma}{Lemma}

\newcommand\C{\hbox{$\mit I$\kern-.7em$\mit C$}}
\newcommand\R{\hbox{$\mit I$\kern-.6em$\mit R$}}

\begin{document}

\title{Quantum Resource Theory for Charge-Parity-Time Inversion}

\author{Michael Skotiniotis}
\email[]{michael.skotiniotis@uibk.ac.at}
\affiliation{Institut f\"ur Theoretische Physik, Universit\"at Innsbruck, Technikerstr.\ 25, A-6020 Innsbruck, Austria}
\author{Borzu Toloui}
\affiliation{Department of Physics, Haverford College, 370 Lancaster Avenue, Haverford, PA 19041, United States}
\author{Ian T. Durham}
\email[]{idurham@anselm.edu}
\affiliation{Department of Physics, Saint Anselm College, Manchester, NH 03102}
\author{Barry C. Sanders}
\email[]{sandersb@ucalgary.ca}
\affiliation{Institute for Quantum Science and Technology,
	University of Calgary, Alberta T2N 1N4, Canada}
\affiliation{Program in Quantum Information Science, Canadian Institute for Advanced Research, Toronto, Ontario M5G 1Z8, Canada}
\date{\today}

\begin{abstract}
We develop a complete resource theory of charge-parity-time (CPT) inversion symmetry for both massive and massless relativistic particles of arbitrary spin. We show that a unitary representation of CPT can be consistently constructed for all spins and develop the resource theory associated with CPT super-selection, thereby identifying and quantifying the resources required to lift the super-selection rule.
\end{abstract} 


\maketitle

\section{Introduction}
\label{sec:intro}

Due solely to Lorentz invariance and the linearity of quantum mechanics, physical laws are invariant under charge,
parity, and time inversion, 
thus making charge-pariy-time  (CPT) symmetry one of the most fundamental symmetries in physics~\cite{Sch51,Lud54,Pau55}.  
Experimental evidence overwhelmingly supports the notion that physical laws are Lorentz invariant~\cite{Gre02} and thus also CPT invariant~\cite{KR11}.

Symmetry in quantum mechanics is mathematically described by a group~$G$ of transformations.
This symmetry implies the presence of invariant states which lead to restrictions on the allowed set of states for the quantum mechanical system.
Specifically, these restrictions give rise to a super-selection rule (SSR)~\cite{Wick:1952ys},
which forbids coherent superpositions of certain quantum states.

SSRs are formally equivalent to the lack of a \textit{quantum} frame of reference~\cite{BRS07}.
Examples of SSRs include those associated with charge~\cite{Wick:1952ys,AS67},
orientation~\cite{BRS03},
chirality~\cite{Wick:1952ys,Collins:2005uq,Gisin:2004fk},
and phase~\cite{AS67,Rudolph:2001uq,SBRK03}.
SSRs may be circumvented by employing the reference frame resource known as frameness
which quantifies the degree for which a given resource is inherent in a shared reference frame.
Such frameness resources correspond
to states that are asymmetric with respect to~$G$~\cite{Gour:2008vn}.

The fundamental nature of CPT symmetry implies the existence of invariant states,
which in turn imply the existence of an SSR.
As with other SSRs,
the CPT-SSR can be circumvented by consuming appropriate frameness resources.
The CPT-SSR thus corresponds to the lack of a quantum reference frame between matter and antimatter.
By treating CPT as an indecomposable symmetry,
in accordance with the Feynman-Stueckelberg interpretation of anti-matter~\cite{Stueckelberg:1942vn,Fey48,Fey49a,Fey49b},
 CPT-SSR and frameness have been analyzed for the cases of massive and massless spin
 $s=0,\;\frac{1}{2}$,
 and $1$ particles as well as for Majorana fermions and Weyl spinors~\cite{STDS13}.

In this treatment,
CPT symmetry was represented by a matter-to-antimatter unitary operator with different forms depending on~$s$ and whether a given particle possesses mass or not,
but the analysis was restrictive in that it pertained only to low-spin cases with definite three-momentum~$\bm p$.

Here we develop a fully general theory of CPT frameness beyond the
 spin and momentum restrictions of the previous work~\cite{STDS13}.
Specifically we employ the Bargmann-Wigner equations~\cite{BW48} to construct 
unitary CPT maps for generalized relativistic Dirac equations with arbitrary~$s$ and 
extend the unitary operator domain from a single value of thee-momentum~$\bm p$
to the full continuum~$\bm{p}\in\mathbb{R}^3$. 
Furthermore, we show that CPT requires a unitary representation
to yield a consistent, unambiguous resource theory.
Our general theory of CPT frameness uses an information-theoretic operational measure
to distinguish states without frameness resources from those with these resources.

The outline of the paper is as follows.  In Sec.~\ref{sec:SSRs} we review the restrictions arising from a general SSR associated with a symmetry group $G$ of transformations and the equivalent resource theory.  We also prove that in order for any resource theory to be consistent and unambiguous in the classification of its resources, the representation of $G$ must necessarily be unitary.

In Sec.~\ref{sec:unitarity} we make use of the Bargmann-Wigner equations and thereby provide an algorithmic procedure for constructing unitary representations of the CPT operator for massive (Sec.~\ref{CPTrqm}) and massless (Sec.~\ref{CPTrql}) relativistic particles of spin $s>1$ and fixed three-momentum~$\bm p$.
Our procedure is prescriptive and allows one to identify CPT invariant and non-invariant states for all relativistic particles of  spin $s>1$. In Sec.~\ref{uniextend} we prove that there exists a unitary extension of the CPT operator for relativistic particles of arbitrary spin and general three-momentum~$\bm p$. 

In Sec.~\ref{sec:resourcetheory} we develop the full resource theory for CPT-SSR, distinguish between non-resource and resource states, and establish a hierarchy of resources for CPT-SSR by quantifying the amount of information such resources carry about the requisite matter-antimatter reference frame. In Sec.~\ref{sec:quantuminfo} we demonstrate that quantum information processing can be performed despite CPT-SSR restrictions.  Finally, Sec.~\ref{sec:conc} contains our summary and conclusions.

\section{\label{sec:SSRs} Superselection and quantum frameness}

In this section we briefly summarize the resource theory due to restrictions imposed by the lack of a frame of reference 
associated with some group~$G$~\cite{BRS07}.
In particular, we provide the necessary mathematical background required to formulate such a resource theory and identify  resources and non-resources for such a theory. 

SSRs were introduced as axiomatic restrictions to quantum theory, 
namely forbidding certain coherent superpositions of quantum states or measurements~\cite{Wick:1952ys,WickWW:70}. 
For~$C$---a conserved quantity represented by Hermitian operator~$\hat{C}$---and for~$\{\ket{c_i}\}$---the corresponding eigenbasis for the Hilbert space~$\mathscr{H}$---a $C$-based SSR states that, for all observables $\mathcal{O}$, 
\begin{equation}
	\bra{c_i}\mathcal{O}\ket{c_j}=0.
\end{equation}
Consequently, any coherent superposition
\begin{equation}
	\ket{\phi}=\alpha\ket{c_i}+\beta\ket{c_j}\in\mathscr{H}
\end{equation}
is operationally indistinguishable from the mixed state
\begin{equation}
	\rho=|\alpha|^2\proj{c_i}+|\beta|^2\proj{c_j}\in\mathcal{S}(\mathscr{H})
\end{equation}
for~$\mathcal{S}(\mathscr{H})$, the set of trace-class positive operators corresponding to physical states~\cite{CMNS:98}.
Thus, the only states $\rho\in\mathcal{S}(\mathscr{H})$ that can be prepared in the presence of a $C$-based SSR are those for which~\cite{CMNS:98}
\begin{equation}
	[\hat{C},\rho]=0.
\end{equation}

As coherent superpositions of different eigenstates of $\hat C$ can neither be observed nor prepared, the Hilbert space $\mathscr{H}$ of any quantum system subject to an SSR can be conveniently written as  
\begin{equation}
	\mathscr{H}\cong\bigoplus_c \mathscr{H}^{(c)};
	\mathscr{H}^{(c)}:=\{\ket{c_i};\, \hat{C}\ket{c_i}=c\ket{c_i}\}
\label{Eq:blockdiagonal}
\end{equation}
with $\mathscr{H}^{(c)}$ the eigenspaces, or charge sectors,
corresponding to distinct eigenspaces of the operator $\hat{C}$.  
In addition, a SSR also imposes restrictions on the types of operations that can be performed. Specifically, the only allowed unitary transformations, $U$, that can be performed under an SSR are those that satisfy~\cite{BRS07}
\begin{equation}
	[U,\hat{C}]=0.
\end{equation}
We now show that the restrictions imposed by the lack of a requisite frame of reference are formally equivalent to an SSR.  

A reference frame is a physical system whose degrees of freedom possess an inherent asymmetry with respect to a particular group of transformations.  Hence, lacking a reference frame for a particular degree of freedom is tantamount to having a symmetry with respect to a group of transformations which, by Noether's theorem~\cite{Noe18,*Tavel:71},
implies that there exists a corresponding conserved quantity.

In quantum frameness theory,
the SSR restriction arises due to action by the group of transformations~$G$
of the reference frame~\cite{BRS07}. 
The Hilbert space~$\mathscr{H}$ carries a representation~$T$ of~$G$. 
By a representation
we refer to the mapping $T: G\rightarrow \mathcal{S}(\mathscr{H})$,
which preserves the group structure:
\begin{equation}
\label{Eq:rep}
	T\left(g^{-1}\right)
		=T^{-1}(g),\;
	T(g g')
		=T(g) T(g'),\;
	\forall g,g'\in G.
\end{equation}
The Hilbert space carries a representation $T$ of $G$,  
\begin{equation}
	\mathscr{H}\cong\bigoplus_{\lambda}\mathscr{H}^{(\lambda)},
\label{Eq:irrep}
\end{equation} 
with~$\lambda$ denoting irreducible representations~$T^{(\lambda)}$ present in~$T$,
and $\mathscr{H}^{(\lambda)}$ denoting the irreducible subspaces that carry $T^{(\lambda)}$.
Equation~\eqref{Eq:irrep} is identical to Eq.~\eqref{Eq:blockdiagonal},
with~$\lambda$ denoting the value of the conserved quantity associated with the requisite frame of reference.

When lacking a frame of reference, the only states that can be prepared are those that satisfy  
\begin{equation}
\label{Eq:nonresources}
T(g) \rho T^{-1}(g)=\rho, \;\;\; \forall g \in G. 
\end{equation}
Quantum states that obey Eq.~\eqref{Eq:nonresources} are called $G$ invariant and constitute the set of non-resource states.

Any state not satisfying Eq.~\eqref{Eq:nonresources} possesses an inherent asymmetry with respect to the group of transformations $\{T(g):\,g\in G\}$ and is therefore a ``frameness'' (or ``asymmetry'') resource~\cite{BRS07}.
Similarly, the only allowable state transformations are those that commute with every element of the representation $T$.
Such transformations are called $G$ covariant maps~\cite{BRS07}.
A map
\begin{equation}
	\cE: \mathcal{S}(\mathscr{H})\rightarrow\mathcal{S}(\mathscr{H})
\end{equation}
is $G$-covariant if  
\begin{equation}
\label{eq:gcovariant}
T(g) \cE(\rho) T^{-1}(g)=\cE\left(T(g) \rho T^{-1}(g)\right)\forall\rho\in\mathcal{S}(\mathscr{H}).
\end{equation}

A resource theory's primary objective is to identify resources and to separate them from non-resources 
subject to a set of SSR restrictions.
Resources are distinguished from non-resources by the fact
that resources cannot be generated from non-resources via SSR-restricted operations~\cite{Gour:2008vn}.
Non-resources can, however, be generated from resources under SSR restrictions.
Thus, the restriction on transformations creates two categories of states: non-resources that only undergo reversible transformations among each other; and resources that can undergo irreversible transformations in a sense being consumed, thus rendering them non-resources~\cite{BRS07,Gour:2008vn}. 

Consistency of a resource theory implies that a $G$-invariant state cannot be transformed via $G$-covariant maps to a state that is not $G$-invariant~\cite{Gour:2008vn}.
We now proceed to show that a resource theory 
employing a unitary representation~$T$ of~$G$ is consistent.
Furthermore the distinction between resource states and non-resource states is unambiguous and always completely consistent.

In quantum mechanics, a system's conservative dynamics is represented by a unitary transformation.
A representation~$T$ of some group~$G$ corresponding to dynamical symmetry is
unitary if and only if each operator $T(g); g\in G$ $\forall G$ is unitary~\cite{CPW02}.
Then Eq.~\eqref{eq:gcovariant} implies that the transformation of a state~$\rho$ by a unitary operator~$U$ satisfies
\begin{equation}
	T(g)U(\rho)=U(\rho) T^{\dagger}(g)\;\forall g\in G.
\end{equation}
As this equality holds for all states~$\rho$, we have
\begin{equation}
	[T(g),U]=0\;\forall g \in G.
\end{equation}

Now suppose the system, initially in a $G$-invariant state $\rho_0$, evolves under 
the action of a Hamiltonian $\hat H$. As the restriction applies to every operation, including infinitesimal ones, the Hamiltonian must also commute with the representation;
i.e.
\begin{equation}
\label{eq:hamiltonianrestriction}
[T(g),\hat{H}]=0 \; \forall\: g \in G.
\end{equation} 
Let us assume that the pure state $\rho_0=\ket{\psi (t_0)}\bra{\psi(t_0)}$ is $G$ invariant.
After a finite time~$t_1$,
the state evolves to 
\begin{equation}
	\rho(t_1)=\text{e}^{- \text{i}\hat{H} (t_1-t_0)}  \rho_0 \text{e}^{+ \text{i}\hat{H} (t_1-t_0)}.
\end{equation}
As $T(g)$ is unitary,
Eq.~\eqref{eq:hamiltonianrestriction} implies
\begin{equation}
\label{eq:c}
	T(g) \text{e}^{-\text{i}\hat{H} (t_1-t_0)} = \text{e}^{-\text{i}\hat{H}(t_1-t_0)} T(g) \;\forall\: g \in G,  
\end{equation}
and 
\begin{align}\nonumber
\label{eq:c2}
	T(g) \rho (t_1) T^{\dagger}(g)
		=& T(g) \text{e}^{-\text{i}\hat{H}(t_1-t_0)} \rho_0
			\text{e}^{+\text{i}\hat{H} (t_1-t_0)} T^{\dagger}(g)\\ \nonumber
		=& \text{e}^{-\text{i}\hat{H}(t_1-t_0)}  T(g) \rho_0 T^{\dagger}(g) \text{e}^{+\text{i}\hat{H}(t_1-t_0)}\\
		=&\rho(t_1).
\end{align}
In other words, the state at $t_1$ remains $G$ invariant and is thus a non-resource state. 

We now consider the case for which the operator~$T(g)$ is anti-unitary for some $g \in G$.
An anti-unitary operator $A: \mathscr{H} \rightarrow \mathscr{H}$ is an anti-linear map such that,
for any two states $\ket{\psi},\,\ket{\phi}\in\mathscr{H}$,
\begin{equation}
\label{eq:anti-unitary}
\left\langle A\psi\right. \ket{A \phi}=\left\langle\psi\right.\ket{\phi}^{*}. 
\end{equation} 
The adjoint $A^{\dagger}$ of an anti-unitary operator is defined as 
\begin{equation}
\label{eq:anti-unitaryadjoint}
\left\langle A^{\dagger}\psi\right. \ket{\phi}:=\left\langle\psi\right.\ket{A\phi}^{*}
\end{equation} 
and is also anti-unitary with $A^{\dagger} A=A A^{\dagger}=\mathds{1}$~\cite{Carinena:1981, Ballentine:1998}. 

Consider then the evolution of an initially pure $G$-invariant state, $\ket{\psi (t_0)}$, under the action
of Hamiltonian~$\hat{H}$, that commutes with the anti-unitary operator $T(g)$. 
A finite time $t_1$ later, the state becomes 
\begin{equation}
	\ket{\psi (t_1)}=\text{e}^{-\text{i}\hat{H}(t_1-t_0)}  \ket{\psi (t_0)}.
\label{eq:psit1-}
\end{equation}
Using Eq.~\eqref{eq:hamiltonianrestriction},
expanding the exponential $\text{e}^{-\text{i}\hat{H}(t_1-t_0)}$ in powers of $\hat{H}$,
and employing the fact that
\begin{equation}
	T(g)\left(i\hat{H}\right)=-\left(i\hat{H}\right)T(g)
\end{equation}
for anti-unitary $T(g)$,
we deduce that
\begin{equation}
	T(g) \text{e}^{-\text{i}\hat{H}(t_1-t_0)} = \text{e}^{+\text{i}\hat{H}(t_1-t_0)} T(g) \; \forall\: g \in G
\end{equation}
and
\begin{align}\nonumber
\label{eq:d2}
	T(g) \ket{\psi (t_1)}
		=&T(g) \text{e}^{-\text{i}\hat{H}(t_1-t_0)} \ket{\psi (t_0)}\\
		=&\text{e}^{+\text{i}\hat{H}(t_1-t_0)}  T(g) \ket{\psi (t_0)}.
\end{align}

As
\begin{equation}
	T(g) \ket{\psi (t_0)}=\ket{\psi (t_0)},
\end{equation}
in order for $\ket{\psi (t_1)}$  to remain $G$-invariant, 
\begin{align}
	\ket{\psi (t_1)}=T(g) \ket{\psi (t_1)}=\text{e}^{+ i H (t_1-t_0)}\ket{\psi (t_0)}
\label{eq:psit1+}
\end{align}
must hold.
Equations~(\ref{eq:psit1-}) and~(\ref{eq:psit1+}) imply 
\begin{equation}
\label{eq:d4}
	\ket{\psi (t_0)}=\text{e}^{- 2\text{i}\hat{H}(t_1-t_0)}  \ket{\psi (t_0)}\;\forall t_1\in\mathbb{R},
\end{equation}
which can only be true for stationary states.

In fact, not all $G$-invariant states can both be stationary and satisfy Eq.~\eqref{eq:d4}.
Thus, choosing the $G$-invariant state $\ket{\psi (t_0)}$ to be non-stationary,
i.e., not an eigenstate of the Hamiltonian,  then~$\ket{\psi(t_0)}$ evolves to a different state,
$\ket{\psi (t_1)}(\neq \ket{\psi (t_0)})$. 
Hence, we have shown that  
\begin{equation}
	T(g) \ket{\psi (t_1)}\neq \ket{\psi (t_1)}
\end{equation}
which implies that $\ket{\psi (t_1)}$ becomes a resource state for any $t_1>t_0$. 

We have therefore demonstrated that,
in the presence of an anti-unitary group operator, the distinction between resource states and non-resource states is violated, and a non-resource state at a given time can evolve  into a resource state  under dynamics that satisfy the restrictions. 
A similar argument holds for mixed states
\begin{align}\nonumber
	T(g) \rho (t_1) T(g)
		=& T(g) \text{e}^{-\text{i}\hat{H}(t_1-t_0)} \rho(t_0) \text{e}^{+\text{i}\hat{H}(t_1-t_0)} T^{\dagger}(g)\\ 
		=&\text{e}^{+\text{i}\hat{H}(t_1-t_0)}  \rho(t_0) \text{e}^{-\text{i}\hat{H}(t_1-t_0)}
				\nonumber\\
		\neq & \rho(t_1).
\label{eq:mixed}
\end{align}

Devising a resource theory associated with an anti-unitary representation faces another related problem. A resource theory classifies resources and orders them from strong to weak~\cite{Gour:2008vn}.
This task of ordering resources by strength entails quantifying resources in a consistent manner.
A highly resourceful state can be transformed to a less resourceful state under the SSR but not the reverse~\cite{BRS07, Gour:2008vn}.

For a consistent resource theory, the relative strength of resources can be quantified in terms of real, positive monotone functions of the states known as resource measures~\cite{Gour:2008vn}. 
For quantum reference frames and their associated symmetry groups,
the relevant measure is the frameness, or asymmetry, measure, $F: \mathcal{S}(\mathscr{H})\rightarrow \mathbb{R}^+$, which must remain the same under any reversible transformation. In particular, 
\begin{equation}
\label{eq:framenessmeasure}
F\left(T(g) \rho T^{\dagger}(g)\right) \equiv F(\rho), \;\;\; \forall g \in G. 
\end{equation}

The value of a measure for a given state in quantum theory cannot depend on the basis in which the state is expressed, as a quantum state is fully described by a unique density operator independent of a choice of basis for the Hilbert space~\cite{NC:00}.
In contrast, if a resource theory exists for an anti-unitary representation,
then any corresponding measure of frameness necessarily  depends on the basis choice.

This basis-dependent property is made clear by observing that any basis change can be expressed as a passive unitary transformation~$U$,
whereas the form of an anti-unitary operator depends on the basis choice.
As the measure~$F$ must satisfy Eq.~\eqref{eq:framenessmeasure},
the form of the measure must also depend on the choice of basis. Such a resource theory cannot be completely consistent as it would imply that the resourcefulness of a state is ambiguous 
since it can change due to a basis transformation, which is not in itself a meaningful physical operation.

As an example, consider the case of a reference frame for the direction of time~\cite{Gour:2009ly}.
The relevant group of transformations is isomorphic to the finite group of two elements~$\mathbb{Z}_2$
with the representation corresponding to time-reversal symmetry which is necessarily anti-unitary.
For the~$\mathbb{Z}_2$ class of transformations,
resource states are qubits (two-level systems),
and the corresponding frameness measure is~\cite{Gour:2009ly}
\begin{equation}
\label{eq:tau}
	\tau(\ket{\psi})=1-\left|\bra{\psi^*}\left.\psi\right>\right|
\end{equation}
with~$\ket{\psi^*}$ denoting the state obtained by conjugating all coefficients for the state $\ket{\psi}$.

For the anti-unitary representation,
the frameness measure given by Eq.~\eqref{eq:tau} is clearly basis dependent.
In particular, for a given $\left\{\ket{0}, \ket{1}\right\}$ basis, where
\begin{equation}
	\ket{\psi}=\psi_0 \ket{0}+\psi_1 \ket{1},
\end{equation}
the value of the measure is equal to 
\begin{equation}
	\tau(\ket{\psi})=1-\left|\psi_0^2+\psi_1^2\right|.
\end{equation}
Therefore, we have to conclude that such a resource theory is not consistent and that, in general, a consistent resource theory cannot be developed for the case that~$\mathscr{H}$ carries an anti-unitary representation of~$G$.

In the next section we construct a unitary representation of the CPT operator for massive and massless  relativistic quantum mechanical systems of arbitrary spin $s>1$.

\section{\label{sec:unitarity} Unitary representation for CPT}

As anti-unitary representations are problematic for frameness,
in this section we devise a procedure for constructing unitary representations of the CPT operator for relativistic quantum systems of arbitrary spin $s>1$, extending the construction of~\cite{STDS13}.
Our algorithm comprises the following three steps:
\begin{enumerate}
	\item[(i)] the construction of relativistically covariant equations describing massive and massless particles of arbitrary spin; 
	\item[(ii)] a state space of solutions for these equations; and
	\item[(iii)] the construction of a unitary CPT operator acting on the state space of solutions.
\end{enumerate}
Steps~(i) and~(ii) for the case of massive particles of spins $s>1$ are given in Sec.~\ref{CPTrqm} and,
for massless particles,
in Sec.~\ref{CPTrql}.
Step~(iii) has been shown previously only for a subspace of the entire state space of solutions,
namely for all states with a fixed $\bm{p}\in\mathbb{R}^3$~\cite{STDS13}.
 In Sec.~\ref{uniextend}, we show that the CPT operator can be extended to the entire state space of 
solutions in such a way that unitarity is preserved.

\subsection{Basis States and Transformations}
\label{subsec:basis}

In the Feynman-Stueckelberg interpretation~\cite{Fey49a} the image of a particle with mass $m$, spin $s$, linear three-momentum~$\bm{p}$, and energy
\begin{equation}
\label{eq:FeynmanStueckelberg}
	E=\sqrt{|\bm{p}|^2c^2+(mc^2)^2}
\end{equation}
under the action of $CPT$ is an antiparticle of the same mass and energy with its spin and three-momentum reversed.
The particle's internal degrees of freedom,
such as electric charge, baryon number, and lepton number,
are inverted.

In terms of the universally conserved internal symmetries, the total internal quantum number is
\begin{equation}
\label{eq:u}
	u := Q + (B-L)
\end{equation}
with~$Q$ the total electric charge and $B-L$ the difference between total baryon number~$B$ and total lepton number~$L$.  We note that, whereas in some theories B and L are not individually conserved, their difference (B-L),  known as the chiral anomaly, is conserved~\cite{Bell:1969uq}.  
As~$m$ and~$E$ are CPT invariant, the particle basis state~$\ket{u,s,\bm p}$ transforms as
\begin{equation}
	CPT\ket{u,s,\bm{p}}=\text{e}^{\text{i}\theta^{\text{CPT}}_{u,s,p}} \ket{-u,-s,-\bm{p}}
\label{eq:CPTstate}
\end{equation}
with
\begin{equation}
	\theta^{\text{CPT}}_{u,s,p}=\theta^{\text{CPT}}_{-u,-s,-p}\in[0,2\pi)
\end{equation}
a global phase.

The state~$\ket{u,s,\bm{p}}$ is not a valid Hilbert-space state as it is not normalizable for $\bm{p}\in\mathbb{R}^3$
but is well defined as a distribution in the distribution space~$\Phi^*$,
which is dual to the nuclear space of test functions~$\Phi$ 
in the Gel'fand triple (or rigged Hilbert space)~\cite{GelfandVilenkin64}
\begin{equation}
	(\Phi,\mathscr{H},\Phi^*).
\end{equation}
Observables are expressed as complex-valued functionals of test functions and distributions.
Here we employ the Dirac adjoint representation to ensure covariance and unitarity throughout~\cite{Gre00}. 

To construct unitary projective representations of $\{\mathds{1},\, CPT\}$, we proceed as follows.
Consider the set of operators
\begin{equation}
	\{\mathds{1},\,C,\,PT,\,CPT\},
\end{equation}
which under composition form a unitary projective representation of the (Abelian) Klein four-group $\mathbb{Z}_2\times\mathbb{Z}_2$.  The action of $\{\mathds{1},\,C,\,PT,\,CPT\}$ with respect to the states~$\ket{u,s,\bm{p}}$
that span the space of distributions~$\Phi^*$ is given by
\begin{align}
\label{eq:subbasis}
	\{\ket{u,s,\bm{p}},
	\text{e}^{\text{i}\theta^{\text{PT}}_{u,s,p}}\ket{u,-s,-\bm{p}}=:&\;PT \ket{u,s,\bm{p}},\nonumber\\
	\text{e}^{\text{i}\theta^{\text{C}}_{u,s,p}}\ket{-u,s,\bm{p}}=:&\;C \ket{u,s,\bm{p}},\nonumber\\
	\text{e}^{\text{i}\theta^{\text{CPT}}_{u,s,p}}\ket{-u,-s,-\bm{p}}=:&\;CPT \ket{u,s,\bm{p}}\}
\end{align}
with $\theta^{\text{CPT}}_{u,s,p}=\theta^{\text{PT}}_{u,s,p}+\theta^{\text{C}}_{u,s,p}$.  Now if $T$ is a representation of a group $G$, then it is also a representation of any subgroup, $S$, of $G$.  Moreover, if $T$ is an irreducible representation of $G$, then the restriction of $T$ to $S$ yields, in general, a reducible representation.
 Indeed, the representation in Eq.~\eqref{eq:subbasis} is a projective representation of the subgroup $\{\mathds{1}, CPT\}$, which is equivalent to~$\mathbb{Z}_2$.

The states
\begin{equation}
\label{eq:usp}
	\{\ket{\pm u,\pm s,\pm\bm p}\}
\end{equation}
are the solutions to a relativistically covariant differential equation that describes a quantum-mechanical system with internal quantum number $u$, spin $s$, and three-momentum $\bm p$.  
Whereas explicit equations exist for low values of $s$, such as the Klein-Gordon equation for $s=0$, the Dirac equation for $s=1/2$, and the Weinberg-Shay-Good equation for $s=1$ {\it inter alia}~\cite{Gre00},
a universally accepted, single-particle equation does not exist for higher-spin particles.  

The most widely employed method of modeling free, massive particles of arbitrary spin $s$ is he use of the Bargmann-Wigner equations~\cite{BW48}.  We now show how these equations are used to construct the relativistic equations and corresponding state space of solutions for the case of massive particles of spin $s>1$ and fixed three-momentum $\bm p$.

\subsection{Massive relativistic particles of arbitrary spin}
\label{CPTrqm}

The Bargmann-Wigner equations model free massive particles of arbitrary spin $s$ as being composed from $2s$  spin-\textonehalf~``primitives''~\cite{BW48}.  The idea behind the Bargmann-Wigner construction can be more easily understood by considering the same situation in non-relativistic quantum mechanics.

Consider the case of $2s$ spin-\textonehalf~systems in standard non-relativistic quantum mechanics.  One can  describe any state of such a system in two equivalent ways.
One description is as a tensor product of states of each of the $2s$ systems in the uncoupled basis.
The other description uses the standard rules for addition of angular momenta
to construct the coupled basis $\{\ket{J,m,\alpha}\}$.
Here $J_0\leq J\leq s$ is the total angular momentum of the $2s$ systems
(with either~$J_0=0$ if~$s$ is an integer or $J_0=1/2$ otherwise),
$-J/2\leq m\leq J/2$ is the projection of the total spin of the $2s$ systems on the $z$ axis,
and $\alpha$ is the degeneracy (multiplicity) index that indicates the number of ways $2s$ spin-\textonehalf~systems can be coupled to produce a single spin-$J$ system.
The Clebsch-Gordan transform connects these coupled and uncoupled bases~\cite{BCH06}.

In the coupled basis description of the $2s$ quantum systems, the subspace with $J=s$, known as the \defn{totally symmetric subspace}, is a $2s+1$-dimensional space, with orthonormal basis
\begin{equation}
	\{\ket{s,M}\}_{M=-s}^{s}\,,
\end{equation}
which can be thought of as the state space of a spin-$s$ particle.
If each spin-\textonehalf~primitive is a positive energy solution of the Dirac equation
(i.e.,~the spin-\textonehalf~primitive is described by the state $\ket{u,\pm 1/2,\bm{p}=\bm{0}}$ where $\bm{p}=\bm{0}$ implies that the solution is with respect to the particle's rest frame),
the symmetric subspace generated by such $2s$ systems corresponds to the positive energy solution of a particle with spin $s$.
In other words
\begin{equation}
	\ket{u, M, \bm{p}=\bm{0}}, \, M\in(-s,\ldots,s)
\end{equation}
where again the solutions are with respect to the rest frame of the spin-\textonehalf~primitives.  By a similar argument the totally symmetric subspace of $2s$ spin-\textonehalf~negative-energy solutions of the Dirac equation corresponds to the negative energy solution of a particle with spin $s$.  

The Bargmann-Wigner construction utilizes the idea outlined above to construct both equations and solutions that describe relativistic particles of arbitrary spin. As each spin-\textonehalf~primitive satisfies the Dirac equation, a spin-$s$ particle satisfies a set of such $2s$ individually indexed Dirac equations. Furthermore, as each spin-\textonehalf~primitive is Lorentz covariant
(by virtue of the Dirac equation being Lorentz covariant),
the spin-$s$ particle also satisfies Lorentz covariance provided the $2s$-fold tensor product representation of the Lorentz group is restricted to the permutationally symmetric subspaces of positive- and negative-energy solutions, respectively. 
 
The Bargmann-Wigner construction leads to a set of $2s$ individually indexed Dirac equations, one for each spin-\textonehalf~primitive, whose positive- and negative-energy solutions describe particles and antiparticles of spin $s$.  Using Eq.~\eqref{eq:subbasis} we proceed to construct a unitary representation of the Klein group which, upon restriction to the subgroup $\{\mathds{1},\, CPT\}$,
yields a unitary representation of the CPT operator.

Specifically, for a spin-$s$ system the unitary CPT operator is a $4(2s+1)\times 4(2s+1)$ matrix with $1s$ in its anti-diagonal. We illustrate our method with the simplest case of constructing a relativistic, massive spin-1 particle out of two massive spin-\textonehalf~primitives.  The latter satisfy the Dirac equation
\begin{equation}
\label{eq:Dirac}
	\left(i\hbar\gamma^\mu\partial_\mu+mc\right)\psi=0
\end{equation}
for Dirac matrices
\begin{equation}
	\gamma^0=\begin{pmatrix}\mathds{1}&0\\0&-\mathds{1}\end{pmatrix},
	\gamma^j=\begin{pmatrix}0&\sigma^j\\-\sigma^j&0\end{pmatrix}
\label{eq9}
\end{equation}
with~$\sigma^j|_{j\in(1,2,3)}$ the Pauli matrices.

For a given three-momentum $\bm p$ the solutions to Eq.~\eqref{eq:Dirac} form an eight-dimensional space spanned by 
\begin{equation}
\label{eq:8d}
	\{\ket{\pm u_{1/2},\pm1/2,\pm\bm{p}}\},
\end{equation}
where the states $\{\ket{u_{1/2},\pm 1/2, \pm \bm p}\}$ correspond to particle states and $\{\ket{-u_{1/2},\pm 1/2, \pm \bm p}\}$ correspond to antiparticle sates.  Using the Bargmann-Wigner construction the state space of a massive (anti-)particle of spin $1$ of definite three-momentum $\bm p$ is spanned by
\begin{widetext}
\begin{align}\nonumber
	\ket{\pm u_1,1,\pm\bm p}
		:=&\ket{\pm u_{1/2},1/2,\pm\bm p}\otimes \ket{\pm u_{1/2},1/2,\pm\bm p}\nonumber\\
	\ket{\pm u_1,0,\pm\bm p}
		:=&\frac{1}{\sqrt{2}}\left(\ket{\pm u_{1/2},1/2,\pm\bm p}\otimes \ket{\pm u_{1/2},-1/2,\pm\bm p}
			+\ket{\pm u_{1/2},-1/2,\pm\bm p}\otimes \ket{\pm u_{1/2},1/2,\pm\bm p}\right)\\ 
	\ket{\pm u_1,-1,\pm\bm p}
		:=&\ket{\pm u_{1/2},-1/2,\pm\bm p}\otimes \ket{\pm u_{1/2},-1/2,\pm\bm p}.\nonumber
\label{eq:spin-1}
\end{align}
\end{widetext}
Note that states that are linear superpositions of $\bm p$ and $-\bm p$ are forbidden as such states describe a spin-$1$ particle of indefinite three-momentum.

The CPT transformation acting on the state space describing the solutions of massive, relativistic, spin-1 systems is represented as a~$4(2\cdot 1+1)\times 4(2\cdot 1+1)=12\times 12$ matrix  with ones on the anti-diagonal. As overall phases can be ignored the $12\times 12$ matrices $\{\mathds{1}, CPT\}$ form a \defn{projective representation} of $\mathbb{Z}_2$~\cite{CPW02}.

\subsection{\label{CPTrql} State space for massless particles of arbitrary spin}

We now carry out the Bargmann-Wigner construction for the case of relativistic massless particles. As we show below, for massless particles, chirality is a Lorentz invariant quantity.  However, as the chirality of the massless spin-$s$ particle, as well as the chirality of each of the constituent, massless spin-\textonehalf~primitives must be Lorentz-invariant, we show that the state space of a massless particle is eight-dimensional irrespective of its spin.

Unlike for massive particles, no frame of reference exists such that $\bm{p}=\bm{0}$ for massless particles. Instead we choose the $z$ axis of the particle to be co-linear with the direction of the momentum such that
\begin{equation}
	(p^{\mu})=(p^{0},0,0,p)=(p^{0},\bm p).
\label{eq:MLmom}
\end{equation}  
The Dirac equation for a massless spin-\textonehalf~particle whose momentum is given by Eq.~\eqref{eq:MLmom} reads 
\begin{equation}
	\text{i}\hbar\gamma^0\gamma^3\partial_0\partial_3\psi=0,
\label{Diracmassless}
\end{equation}
and the state space of solutions is again spanned by $\{\ket{\pm u_{1/2},\pm1/2,\pm\bm{p}}\}$. The fact that the particles are massless means that if $\psi$ is a solution of Eq.~\eqref{Diracmassless},
so is $\gamma^5\psi$,
where $\gamma^5:=\text{i}\gamma^0\gamma^1\gamma^2\gamma^3$ is the chiral operator
and can be shown to satisfy
\begin{equation}
	\left[\gamma^5,\sigma^{\alpha\beta}\right]=0\,\forall\, \alpha,\,\beta,\;
	\sigma^{\alpha\beta}=\frac{i}{2}[\gamma^\alpha,\gamma^\beta].
\label{eq:Dirac_invariant}
\end{equation}
If the algebra of Dirac operators is irreducible, then Schur's lemmas imply that the only matrix that commutes with every Dirac operator is a multiple of the identity~\cite{CPW02}.

As $\gamma^5$ is not a multiple of the identity, the Dirac operators can be further reduced into their irreducible components. Consequently, the state space of solutions of Eq.~\eqref{Diracmassless} can also be decomposed into irreducible subspaces.
The irreducible subspaces of the state space of solutions of Eq.~\eqref{Diracmassless} are the eigenspaces of the projection operators 
\begin{align}\nonumber
	\psi_{L}&=1/2\left(\mathds{1}-\gamma^5\right),\\
	\psi_{R}&=1/2\left(\mathds{1}+\gamma^5\right).
\label{eq:chirality}
\end{align}
The $\psi_L$ and $\psi_R$ eigenstates of Eq.~\eqref{eq:chirality} are known as Weyl spinors
with left- and right-handed chirality respectively.
Each Weyl spinor obeys Eq.~\eqref{Diracmassless}. 
Assigning $\partial_i=\bm {p_i}$, $\gamma^0\gamma^i=\gamma^5\Sigma^i$, with
\begin{equation}
	\Sigma^i\equiv\frac{\text{i}\epsilon^{ijk}}{2}[\gamma^j,\gamma^k],
\end{equation}
and some algebra Eq.~\eqref{Diracmassless} for each Weyl spinor reads
\begin{equation}
\text{i}\hbar\left(\gamma_5-\frac{\bm\Sigma\cdot\bm p}{|\bm p|}\right)\psi_{L(R)}=0.
\label{eq:helicity=chirality}
\end{equation}
The quantity $\frac{\bm\Sigma\cdot\bm p}{|\bm p|}$ is known as the \defn{helicity} of the spinor.
For massless spin-\textonehalf~particles helicity is equal to chirality.  

Thus, the solutions to the Dirac equation for a massless spin-\textonehalf~system can be described as follows. A positive-energy (negative-energy) solution, corresponding to a particle (antiparticle), can have either positive or negative helicity. The latter requirement fixes the allowable solutions to
\begin{subequations}
\begin{align} 
&\ket{u,1/2,\bm p},\quad \ket{u,-1/2,-\bm p} \label{eq:massless_part_pos_hel} \\ 
&\ket{u,-1/2,\bm p},\quad \ket{u,1/2,-\bm p} \label{eq:massless_part_neg_hel} \\ 
&\ket{-u,1/2,\bm p},\quad \ket{-u,-1/2,-\bm p} \label{eq:massless_anti_pos_hel} \\ 
&\ket{-u,-1/2,\bm p},\quad \ket{-u,1/2,-\bm p} \label{eq:massless_anti_neg_hel}.
\end{align}
\end{subequations}
Equation~\eqref{eq:massless_part_pos_hel} describes a massless spin-\textonehalf~particle with positive helicity, whereas Eq.~\eqref{eq:massless_part_neg_hel} describes a a massless spin-\textonehalf~particle with negative helicity.  Likewise, Eqs.~\eqref{eq:massless_anti_pos_hel} and~\eqref{eq:massless_anti_neg_hel} describe an antiparticle with positive and negative helicity respectively. 

Now we use the Bargmann-Wigner construction to construct a massless spin-$s$ system out of $2s$ massless spin-\textonehalf~systems. Using the Bargmann-Wigner construction, a massless spin-$s$ particle and antiparticle are described by the symmetric subspace of $2s$ massless spin-\textonehalf~particles and antiparticles respectively.

Each massless spin-\textonehalf~primitive has the same rest frame. We require that valid states be definite eigenstates of $\bm p$, and
\begin{equation}
	S_z=\sum_{i}\sigma_z^{(i)},
\end{equation}
as well as of the total chirality operator
\begin{equation}
	\Gamma^5=\sum_i\gamma^5_i.
\end{equation}
However, as chirality for massless particles is a Lorentz invariant quantity, we require that the chirality of each individual constituent, be preserved.  This latter requirement restricts the allowable states for a massless, spin-$s$ system to those for which $S_z=\pm s$ as the next lemma shows.

\begin{lemma}
Let
\begin{equation}
	\{\ket{u_{1/2},m,p_z}\}_{m=-1/2}^{1/2}
\end{equation}
be the state space of a massless spin-\textonehalf~primitive, with $\bm p=p_z$ and $m$ the angular momentum projection onto the $z$ axis. Also let
\begin{equation}
\label{eq:us2}
	\{\ket{u_{s},M,p_z}\}_{M=-s}^{s}
\end{equation}
represent the totally symmetric subspace of $2s$ such massless spin-\textonehalf~primitives.  The only states in the totally symmetric subspace for which exist stationary  states of the local chirality operator $\gamma^5_i$ for all $i$ are the states
\begin{equation}
\label{eq:us2pms}
	\left\{\ket{u_{s},\pm s,p_z}\right\}.
\end{equation}
\label{lem:1}
\end{lemma}

\begin{proof}
We begin by first showing that the states of Eq.~\eqref{eq:us2} have definite total chirality, i.e.,~are eigenstates of $\Gamma^5$.
We then show that the only non-entangled states in the set of states of Eq.~\eqref{eq:us2} are the states with $M=\pm s$ and that these states are the only states which, upon tracing all but the $i$th system,  result in an eigenstate of the local chirality operator $\gamma^5_i$ for all $i$. 

From the theory of angular momentum~\cite{Sakurai:1994fk} the symmetric state $\ket{u_{s},M,p_z}$ with $M=s-2k$ can be written as
\begin{align}\nonumber
	\ket{u_{s},M,p_z}=&\frac{1}{\sqrt{\binom{s}{k}}}\sum_{\pi\in S_s}\pi
		\left(\ket{u_{1/2},1/2,p_z}^{\otimes (s-k)}\right.\\
	&\left.\otimes\ket{u_{1/2},-1/2,p_z}^{\otimes k}\right)
\label{eq:dickestates}
\end{align}
with~$\ket{a}^{\otimes 2}:=\ket{a}\otimes\ket{a}$,
and the sum in Eq.~\eqref{eq:dickestates} is over all permutations $\pi\in S_s$ of $2s$ objects that result in a unique re-ordering of the $2s$ systems.

As each summand in Eq.~\eqref{eq:dickestates} contains $s-k$ eigenstates of $\gamma^5_i$ with $+1/2$ eigenvalue, and $k$ eigenstates of $\gamma^5_i$ with $-1/2$ eigenvalue
\begin{equation}
\Gamma^5\ket{u_{s},M,p_z}=M\ket{u_{s},M,p_z}.
\end{equation}
Hence, the states of Eq.~\eqref{eq:us2} are eigenstates of $\Gamma^5$ with eigenvalue $M$.  

We now determine the state of one of the constituent massless spin-\textonehalf~systems by tracing out all the remaining spin-\textonehalf~systems in Eq.~\eqref{eq:dickestates}.  Without loss of generality we may choose to keep the first system in Eq.~\eqref{eq:dickestates}.  For any given $M$ in Eq.~\eqref{eq:dickestates} the first massless system is in the state
$\ket{u_{1/2},1/2, p_z}$ $\binom{s-1}{s-k-1}$ times, and in the state $\ket{u_{1/2},-1/2, p_z}$ $\binom{s-1}{k-1}$ times.  Hence the reduced density matrix, $\rho_1$, for the first massless system can be easily shown to be given by
\begin{align}\nonumber
\rho_1=&\frac{\binom{s-1}{s-k-1}}{\binom{s}{k}}\ketbra{u_{1/2},1/2, p_z}{u_{1/2},1/2, p_z}\\
&+\frac{\binom{s-1}{k-1}}{\binom{s}{k}}\ketbra{u_{1/2},-1/2, p_z}{u_{1/2},-1/2, p_z}.
\label{eq:reduced_state}
\end{align}
Whereas we have computed the reduced density matrix only for the first massless system, Eq.~\eqref{eq:reduced_state} is true for any massless spin-\textonehalf~system as all states we are considering are permutationally symmetric.

Now $\rho_1$ is pure if and only if $k=0$ or $s$, which correspond to $M=\pm s$, or else $\rho_1$ is mixed.  As the entropy of the reduced density matrix of a system is a measure of entanglement, it follows that the states
$\ket{u_{s},\pm M,p_z}$ for $M=\pm s$ are separable, whereas all other states are entangled.  

We now show that the only states that satisfy local invariance of chirality are the states with $M=\pm s$. 
As each massless spin-\textonehalf~system has the same rest frame, and as the chirality of each constituent system is a conserved quantity, it follows that the only allowable states in the symmetric subspace are those for which the reduced states $\rho_i$ are eigenstates of $\gamma^5_i$ with a definite eigenvalue for all $i$. From Eq.~\eqref{eq:reduced_state} this occurs only for the states given by Eq.~\eqref{eq:us2}.  This completes the proof.
\end{proof}

Lemma~\ref{lem:1} shows that there are only two allowable states for a massless spin-$s$ particle of definite momentum $\bm p$.  Similarly a massless spin-$s$ antiparticle of definite three-momentum, $\bm p$, can only possess two possible spin states. Taking into account that each massless particle and antiparticle can have momentum $\bm p$ or $-\bm p$ the space describing all possible valid states of a massless spin-$s$ system using the Bargmann-Wigner construction is eight-dimensional, independent of the spin of the system. The corresponding CPT operator is represented by an $8\times 8$ matrix with ones on the anti-diagonal. These matrices $\{\mathds{1}, CPT\}$ form a projective representation of $\mathbb{Z}_2$.

Let us illustrate our construction for the simplest case of building a massless spin-$1$ system out of two massless spin-\textonehalf~primitives.
Thus, the state space of spin-$1$ particles is spanned by 
\begin{subequations}
\begin{align}
	\ket{u_1,\pm1,\pm \bm p}&:=\ket{u_{1/2},\pm1/2,\pm \bm p}\otimes\ket{u_{1/2},\pm 1/2,\pm \bm p} \label{eq:massless_spin-1_pos_hel} \\
	\ket{u_1,\mp 1,\pm \bm p}&:=\ket{u_{1/2}, \mp1/2,\pm\bm p}\otimes\ket{u_{1/2}, \mp1/2,\pm\bm p}, \label{eq:massless_spin-1_neg_hel}
\end{align} 
\end{subequations}
where Eq.~\eqref{eq:massless_spin-1_pos_hel} describes particles with positive helicity, 
whereas Eq.~\eqref{eq:massless_spin-1_neg_hel} describes particles with negative helicity.  The corresponding states for the antiparticle are obtained by replacing $u_{1/2}$ with $-u_{1/2}$ in Eqs.~\eqref{eq:massless_spin-1_pos_hel} and~\eqref{eq:massless_spin-1_neg_hel}.

Note that a massless spin-$1$ particle can have its spin, which corresponds to polarization, either parallel or anti-parallel to its direction of motion, representing states of positive and negative helicity respectively. The lack of longitudinal polarization arises naturally  in this construction from the fact that the state 
\begin{align}
\ket{u_1,0,\bm p}&:=1/2\left(\ket{u_{1/2},1/2,\bm p}\otimes\ket{u_{1/2},-1/2,\bm p}\right.\nonumber \\
&\left.+\ket{u_{1/2},-1/2,\bm p}\otimes\ket{u_{1/2},1/2,\bm p}\right)\nonumber
\end{align}
is an eigenstate of $\Gamma^5$, but not of $\gamma^5_1$ or $\gamma^5_2$. 

Finally, the $8\times 8$ matrices $\{\mathds{1}, CPT\}$ form a projective representation of $\mathbb{Z}_2$. The eigenstates of CPT physically correspond to linear superpositions of states with opposite helicity.  We remark that the construction given here gives rise to the same state space and CPT operator as the construction in~\cite{STDS13}, where the state space and corresponding CPT operator were obtained by considering the Bia{\l}ynicki-Birula--Sipe  equation~\cite{Bialynicki-Birula:1994fk,Bialynicki-Birula:1995fk,Bialynicki-Birula:1996uq,Sipe:1995uq,Kobe:1999ve,Smith:2007kx,Raymer:2008zr}.

Thus far we have proven that our construction of the CPT operator is unitary 
only on the sub-basis defined by a definite value three-momentum $\bm p$.
In the next sub-section we extend the domain of the CPT operator to the generalized functions of~$\bm p$,
thereby proving that this extension retains unitarity.

\subsection{\label{uniextend} Extending CPT to general three-momentum support}

In this subsection we show that the $CPT$ operator constructed in  Sec.~\ref{subsec:basis} is well defined even with support over a generalized function of~$\bm p$.   
Our strategy is first to consider the action of the CPT operator on the test functions in~$\Phi$ and then to extend the operator to the continuous basis in~$\Phi^*$. We next derive the unitarity of the extension of CPT from the unitarity of its reduction on the test functions. Finally, we show how the projections of the CPT operator onto subspaces of fixed $\bm{p}$ are each unitary but normalized by the delta function.   

The inner product of an arbitrary state $\ket{\phi} \in \Phi$ 
and distribution $\bra{u,s,\bm{p}} \in \Phi^*$ is
\begin{equation}
	\bracket{u,s,\bm{p}}{\phi}=\phi(u,s,\bm p),
\label{eq:dist-action}
\end{equation}
which is a smooth, rapidly decreasing test function such that
\begin{equation}
	\int_{\mathbb{R}^3}\mathrm{d}\bm p\left|\bm p^n \phi(u,s,\bm p)\right|^2<\infty. 
\end{equation}
for all $n=0,1,2,\ldots$~\cite{Bohm:86}. 
The function $\phi(u,s,\bm p)$ is the complex amplitude of finding the physical system in state $u$, with spin $s$, and three-momentum $\bm p$.  
The state $\ket{\phi}$ can be expanded in terms of the continuous basis states $\left\{\ket{u,s,\bm{p}}\right\}$ in the space of distributions $\Phi^{*}$ as 
\begin{align}
	\ket{\phi}
		&=\sum_{u,s}\int\mathrm{d}\bm p \: \phi(u,s,\bm p)\ket{u, s, \bm p}
\label{testfunction}
\end{align}
 
Consider the reduction of the  $(CPT)_{\Phi}$ operator on the nuclear space~$\Phi$ and its action on a test function $\ket{\phi}$.
From Eq.~(\ref{testfunction}), we obtain 
\begin{equation}
	\bra{u,s,\bm{p}}(CPT)_{\Phi}\ket{\phi}=\text{e}^{\text{i}\theta^{\text{CPT}}_{u,s,p}} \phi(-u,-s,-\bm p). 
\label{eq:CPTdagger-sandwich}
\end{equation}
We can extend the CPT operator to the space $\Phi^*$ of distributions via
\begin{align}
\bracket{(CPT)^{\dagger}(u,s,\bm{p})}{\phi}:=\bra{u,s,\bm{p}}(CPT)_{\Phi}\ket{\phi},
\end{align}
which,
together with Eqs.~(\ref{eq:CPTdagger-sandwich}) and~(\ref{eq:dist-action}),
yield the formal extension~\cite{Bohm:86, Madrid02}:
\begin{align}
	\ket{u,s,\bm{p}} (CPT)^{\dagger}= \text{e}^{-\text{i}\theta^{\text{CPT}}_{u,s,p}}\ket{-u,-s,-\bm p}, 
\label{eq:CPT-dagger}
\end{align}
Note that  $(CPT)_{\Phi}$  is continuous and unitary in the nuclear space $\Phi$.  As the latter is dense in $\mathscr H$, the operator can be extended to a unitary operator in $\mathscr H$, and the adjoint of such an operator on $\Phi^*$ has a complete system of generalized eigenvectors in the rigged Hilbert space~\cite{GelfandVilenkin64}. 
The unitarity of $(CPT)_{\Phi}$ on $\Phi$ implies that 
\begin{align}
	(CPT)^{\dagger} \:(CPT)_{\Phi}
		=\mathds{1}_{\Phi}, 
\end{align}
where $\mathds{1}_{\Phi}$ is the identity operator in~$\Phi$. 

Given that the space $\Phi^*$ is a vector space in its own right, we can likewise define the action of the CPT-operator on the kets labeled by $\ket{u,s,\bm{p}}$ from Eq.~(\ref{eq:CPT-dagger}) as 
\begin{align}
	CPT\ket{u,s,\bm{p}}= \text{e}^{+\text{i}\theta^{\text{CPT}}_{u,s,p}}\ket{-u,-s,-\bm p}
\label{formal_identities}
\end{align}
(or, equivalently, we could have followed a similar derivation from the restricted adjoint operator $(CPT)_{\Phi}^{\dagger}$ instead).
Finally, we can restrict the  resulting $CPT$ operator (over $\Phi^*$) to the subspace spanned by basis states with fixed magnitude of momentum equal to $|\bm{p}|$ as 
\begin{align}
	(CPT)_{\bm p}
		:=\sum_{u,u',s,s'} \sum_{\bm{p'}=\pm \bm{p}}&\bra{u,s,\bm{p}} CPT\ket{u',s',\bm{p'}}
				\nonumber \\&
			\times \ketbra{u',s',\bm{p'}}{u,s,\bm{p}}. 
\label{cptrigged}
\end{align} 
As the CPT transformation does not mix test functions evaluated at different $|\bm {p}|$ among each other, and thus cannot be used to create linear superpositions of  basis elements $\ket{u,s,\bm{p}}$ associated with different $|\bm {p}|$, the unitarity of the $CPT$ operator carries over to its reduction $(CPT)_{\bm{p}}$. In particular, we have   
\begin{align}
 (CPT)_{\bm p'} \: (CPT)_{\bm p}^{\dagger} =\delta(p-p')\mathds{1},
 \label{eq:delta-unitary}
\end{align}
and each $(CPT)_{\bm{p}}$ is, by itself, unitary on the subspace spanned by basis states with fixed $|\bm{p}|$.

\section{CPT Resource theory for a CPT frame of reference}
\label{sec:resourcetheory}

In this section we formulate the resource theory associated with the lack of a reference frame for matter and antimatter.  Specifically, we distinguish CPT resource states from non-resource states.  In addition we quantify a physical system's ability to act as a CPT reference frame using an information-theoretic, operational measure that assigns zero frameness to non-resource states and establishes a consistent, unambiguous hierarchy among all resource states.

In Sec.~\ref{sec:SSRs} we showed that an SSR is equivalent to the lack of a reference frame that is associated with a symmetry with respect to a group $G$ of transformations.  If $G=\mathbb{Z}_2$ then, as we showed in Sec.~\ref{sec:unitarity}, the representation $\{\mathds{1},\, CPT\}$ is a unitary representation of $\mathbb{Z}_2$ on the space of distributions $\Phi^*$.

As $CPT$ transforms matter into antimatter and vice versa, the representation $\{\mathds{1},\, CPT\}$ describes the transformations of a reference frame associated with matter and antimatter.  Therefore, a CPT-SSR is associated with the lack of a common reference frame for matter and antimatter.

Due to Schur's lemmas~\cite{CPW02}, unitary representations of finite groups can be fully reduced into their irreducible representations (IRs).  In particular, $\mathbb{Z}_2$ has two one-dimensional IRs given by $\pm$.
As CPT-SSR implies that coherent superpositions between eigenstates of the $CPT$ operator cannot be observed~\cite{Gour:2008vn}, the space of distributions $\Phi^*$ of any system subject to CPT superselection may conveniently be written as 
\begin{equation}
	\Phi^*\cong\bigoplus_{\epsilon\in\{\pm\}}\Phi^{*(\epsilon)},
\label{eq:cptdirectsum}
\end{equation}
with IR label $\epsilon$ denoting the two inequivalent IRs of $\mathbb{Z}_2$,
and $\Phi^{*(\epsilon)}$ denoting the corresponding eigenspaces.

Equation~\eqref{eq:cptdirectsum} is identical to Eq.~\eqref{Eq:irrep} with $\mathscr{H}$ replaced by 
the space of distributions~$\Phi^*$, which is a (continuous) sum of all $\Phi^*_{\bm{p}}$. Thus, the space of distributions is partitioned into $\Phi^{*(\pm)}$
corresponding to the eigenspaces spanned by the distributions of $CPT$ eigenvectors with positive and negative eigenvalues, respectively. 

The states that can be prepared in the absence of a CPT frame of reference, the non-resource states, are test functions that belong in either $\Phi^{(+)}$ or $\Phi^{(-)}$,
which are dual to the spaces $\Phi^{*(+)}$ or $\Phi^{*(-)}$, respectively.
Hence, a linear superposition of eigenstates of $CPT$ is a resource and can be brought, via CPT invariant operations, to the standard form
\begin{equation}
	\ket{\psi}=\sqrt{q_0}\ket{+}+\sqrt{q_1}\ket{-},\,
	q_0\in [0,1],\,
	q_1=1-q_0,
\label{eq:standard}
\end{equation}
with $\ket{\pm}$ arbitrary states from $\Phi^{(\pm)}$.

The important point is that the state in Eq.~(\ref{eq:standard}) is a superposition of two states
chosen from two $\mathbb{Z}_2$ IR labels~$\pm$.
For simplicity we can consider the state as being in a fixed momentum state,
i.e., a plane wave.
As a perfect plane wave is unphysical, a more realistic treatment would have the state
prepared in a wavepacket with support over a continuum of momentum values~$\bm{p}$.

We now introduce a frameness monotone, the alignment rate $R(\psi)$, to quantify the resourcefulness of the state $\ket{\psi}$ in Eq.~\eqref{eq:standard}.
The alignment rate is an information-theoretic, operational measure that quantifies the average amount of classical information about a matter-antimatter reference frame in a reference-frame alignment protocol~\cite{SG12}.  

For reference frame alignment, two parties---Alice and Bob---each possessing their own matter-antimatter reference frame,
seek to align their corresponding reference frames by exchanging  the appropriate quantum-mechanical systems.  
For this purpose,
Alice prepares~$N$ copies of a quantum system in an initial state, $\ket{\psi(g)},\, g\in G$,
which contains information about her reference frame.
She sends~$N$ copies of this system to Bob who performs a measurement relative to his frame of reference.  

Bob's measurement outcome $g'\in G$ serves as a guess for Alice's reference frame.  The success of the protocol is quantified by a suitable figure of merit, $f(g,g')$.  Our goal is to determine the state $\ket{\psi}$ prepared by Alice, and the measurement performed by Bob,
such that $f(g,g')$ is maximized. 

As the symmetry group corresponding to a matter-antimatter reference frame is $\mathbb{Z}_2$, reference frame alignment amounts to Alice and Bob determining which $g\in\mathbb{Z}_2$ relates their corresponding frames of reference.  If the states $\ket{\psi}^{\otimes N}$ and $(CPT\ket{\psi})^{\otimes N}$ are orthogonal, then there exists a measurement that perfectly distinguishes them. Hence, upon performing such a measurement, Bob can infer with certainty which $g\in\mathbb{Z}_2$ relates his and Alice's reference frame.

Greater distinguishability between $\ket{\psi}^{\otimes N}$ and $(CPT\ket{\psi})^{\otimes N}$
implies more classical information can be accessed by Bob about $g\in G$.  
Thus, a figure of merit should be a function that quantifies the amount of information Bob learns about Alice's frame of reference. The alignment rate~\cite{SG12} quantifies the amount of information Bob learns, on average, about Alice's frame of reference per copy of the state $\ket{\psi}$ in the limit that 
Alice transmits asymptotically many copies. 

For the unitary representations of $\mathbb{Z}_2$,
the alignment rate is~\cite{SG12}
\begin{equation}
	R(\psi)=-2\log\left|q_0-q_1\right|.
\label{eq28}
\end{equation}
For non-resource states,
i.e., for the $CPT$ eigenstates with either $q_0$ or $q_1$ equal to zero,
$R=0$ indicating that these states carry zero frameness. Furthermore, for non-resource states,
the alignment rate is a monotonically increasing function and is effectively infinite for states with $q_0=q_1=1/2$. The latter are perfect tokens of a matter-antimatter reference frame as, in this case, the pair of states
$\{\ket{\psi},CPT\ket{\psi}\}$ are perfectly distinguishable by suitable measurements.  Such perfect tokens correspond, for example, to forward propagating particles with internal symmetry, $u$ and spin $s$. 

In the next section we demonstrate that Alice and Bob can even perform quantum information processing in the absence of a shared matter-antimatter reference frame, and without the need to establish such a shared frame of reference.

\section{\label{sec:quantuminfo} Quantum information without a commonly shared matter-antimatter frame of reference}

In this section we show that quantum information processing can be performed despite a CPT-SSR by exploiting the degeneracy of the $\pm1$ eigenvalues of the CPT operator.
Similar to the cases of spin-$0$, $\tfrac{1}{2}$, and~$1$~\cite{STDS13},
we consider the case that Alice prepares a linear superposition of $+1$ eigenstates of the CPT operator for a massive spin-$s$ system, i.e.,
\begin{equation}
	\ket{\psi}=\sum_{j=0}^{2(2s+1)} \alpha_j\ket{+,j,\bm p},
\label{eq:dfs}
\end{equation}
where the states $\{\ket{+,j,\bm p}\}_{j=0}^{2(2s+1)}$ span the $2(2s+1)$-dimensional eigenspace corresponding to the +1 eigenvalue of the CPT operator.  As this state is a $CPT$ eigenstate for all $\alpha_j\in\mathbb{C}$, Bob's state is represented exactly the same as Alice's. Choosing the coefficients $\alpha_j$ appropriately,
Alice can encode $\log_2 2(2s+1)$ logical qubits of information, which Bob can retrieve by an appropriate decoding without having to pre-establish a shared frame of reference for matter and antimatter. Note that if Alice and Bob use massless spin-$s$ systems, then the same protocol as above allows Alice to transmit $\log_2 4=2$ qubits of information regardless of the spin of the massless particle.

The ability to encode and decode information in the manner described in Eq.~\eqref{eq:dfs} arises purely from the degeneracy of the $\pm 1$ eigenvalues of the CPT operator.  Specifically, recall that by Eq.~\eqref{eq:cptdirectsum} the space of distributions $\Phi^*$
can be conveniently decomposed into a direct sum of spaces $\Phi^{*(\pm)}$.
Furthermore, each such subspace can be decomposed into a \defn{virtual} tensor product 
\begin{equation}
	\Phi^{*(\pm)}=\cM^{*(\pm)}\otimes\cN^{*(\pm)},
\label{eq:virtualtensorproduct}
\end{equation}
where $\cM^{*(\pm)}$ is the space upon which the IRs~$T^{(\pm)}$,
present in the representation $T$ of $\mathbb{Z}_2$,
act,
and $\cN^{*(\pm)}$ is the space upon which the trivial (identity) representation of $\mathbb{Z}_2$ acts. Note that the spaces $\cM^{*(\pm)}$ and $\cN^{*(\pm)}$ do not correspond to physical systems, which is why the tensor product in Eq.~\eqref{eq:virtualtensorproduct}  is called a virtual tensor product~\cite{Zan01}.

As~$\mathbb{Z}_2$ has only two IRs, both of which are one dimensional, the dimension of $\cM^{*(\pm)}$ is 1.  
Thus, the dimension of the space $\cN^{*(\pm)}$ is equal to the degeneracy of the $\pm$ eigenvalues of $CPT$.

As~$CPT$ acts trivially on the subspaces $\cN^{*(\pm)}$, such subspaces are known as decoherence-free, or noiseless, subspaces~\cite{ZR97} and have been used extensively in quantum information for constructing error-avoiding codes~\cite{Kempe01}. The amount of quantum information that can be encoded in such decoherence-free subspaces is equal to the logarithm of their dimension.  Hence, as long as the $CPT$ operator contains at least one degenerate eigenvalue in its spectrum, quantum information processing can be performed even with a CPT-SSR.

\section{Conclusion}
\label{sec:conc}

We have studied CPT SSRs and developed a general theory of CPT frameness
that greatly extends our prior work on CPT frameness~\cite{STDS13}.
Our earlier study was restricted to particles with spin $s\leq 1$ and assumed that non-normalizable states
had some precise value  of three-momentum.
Here we employ the Bargmann-Wigner approach to develop a procedure for constructing a
unitary CPT operator for any spin~$s$ and any generalized three-momentum.
We show that this CPT operator is always unitary,
and that its isomorphism to $\mathbb{Z}_2$ leads to single-qubit resource states.

\acknowledgements
B.~T.~thanks I.\ Marvian and P.\ Aniello for helpful discussions.
M.~S.~acknowledges support from the Austrian Science Fund, Grants No.~P24273-N16, and No.~SFB F40-FoQus F4012-N16. 
I.~T.~D.~acknowledges financial support from Foundational Questions Institute (FQXi). B.~C.~S.~acknowledges financial support from Alberta Innovates Technology Futures (AITF) and Natural Sciences and Engineering Research Council of Canada (NSERC).

\bibliography{cpt}
\end{document}